\documentclass[pdflatex,sn-basic]{sn-jnl}

\usepackage{graphicx}%
\usepackage{multirow}%
\usepackage{amsmath,amssymb,amsfonts}%
\usepackage{amsthm}%
\usepackage{mathrsfs}%
\usepackage[title]{appendix}%
\usepackage{xcolor}%
\usepackage{textcomp}%
\usepackage{manyfoot}%
\usepackage{booktabs}%
\usepackage{algorithm}%
\usepackage{algorithmicx}%
\usepackage{algpseudocode}%
\usepackage{listings}%
\usepackage{mathtools}
\usepackage[T1]{fontenc}

\theoremstyle{thmstyleone}%
\newtheorem{theorem}{Theorem}
\newtheorem{proposition}[theorem]{Proposition}%

\newtheorem{lemma}[theorem]{Lemma}
\newtheorem{corollary}[theorem]{Corollary}
\theoremstyle{thmstyletwo}

\newtheorem{assumption}[theorem]{Assumption}
\theoremstyle{remark}
\newtheorem{remark}{Remark}

\theoremstyle{thmstylethree}
\newtheorem{definition}{Definition}

\begin{document}

\title[Identifiability of a dynamical system from a single snapshot]{Identifiability of a linear dynamical reaction-diffusion system with directed interactions from the second-order statistics of a single snapshot}

\author[1]{\fnm{Pirmin} \sur{Schlicke}}\email{pirmin.schlicke@plus.ac.at}

\affil[1]{\orgdiv{Department of Biosciences and Medical Biology}, \orgname{University of Salzburg}, \state{Salzburg}, \country{Austria}}

\abstract{

A single spatial snapshot of a system with multiple interacting populations captures positions but not the underlying dynamics that have led to the observation. This work asks what the stationary second-order spatial statistics of a snapshot can reveal about the structure of the linear reaction-diffusion dynamical model, and in particular about the pairwise directed entries of its interaction operator. The observation is the matrix-valued structure factor $S(k)$ as a one-parameter family in the diffusion-inherited dependence on the wavenumber $\kappa=|k|^2$. The main result is that the directed interaction operator $A$ and the diffusivity matrix $D$ are generically identifiable for any number of interacting populations from $\{S(\kappa)\}$ up to a single global time rescaling.
The findings do not contradict known bounds and the recent result that spatial observational data only determine the Markov equivalence class of the underlying graph. In this work, we identify physical assumptions that lift this barrier. The absolute time scale is not identifiable from a single snapshot of any sample size and requires an external clock. For tumor-immune tissue, the results imply that directed cell-cell interactions are in principle recoverable from one section, that the direction cannot be inferred from the mean densities or a covariance matrix, but from the wavenumber dependence of the cross and autospectra, and that a baseline section determines the interaction structure but not the absolute rates.}


\maketitle

\section{Introduction}\label{sec:intro}

Multiplexed tissue imaging resolves the identity and position of thousands of single cells in one imaged tissue field \citep{keren2018,schurch2020,danenberg2022}, and spatial statistics of the resulting multi-type cell point patterns, in particular pair correlation functions and neighborhood enrichment, are routinely utilized to describe immune cell localization \citep{keren2018,schurch2020,bull2020}. The observation of a fixed tissue section captures a single point in time, whereas the underlying dynamics that have produced it (for example directed cell-cell interactions) are quantities of interest in clinical decision-making. The target of a therapeutic intervention is typically such a directed interaction, for example which cell type kills, recruits or suppresses another. Opposite interaction directions imply different interventions, yet they can produce the same spatial proximity, since both place the two cell types close to each other. The descriptors in routine use, such as composition, neighborhood enrichment and co-localization, summarize the pattern to only a few numbers per cell type or pairs, but not necessarily to interaction directionality. Longitudinal observations of the same tissue would resolve the directionality, but they are commonly not available in oncology, and particularly not in routine care.

Spatially resolved tumor-immune dynamics have been modeled by reaction-diffusion systems \citep{matzavinos2004}, and reviews of the non-spatial models of the same interactions are available \citep{eftimie2011}. In the following, the linearization around a stationary state of such a system is investigated. For a single population of cells it has been shown that the net proliferation rate and the diffusion coefficient of a reaction-diffusion model can be calibrated from a biopsy using the spatial power spectrum of multiplex-immunofluorescence section \citep{pasetto2024}. It can also be achieved from routine histology slides by jointly fitting the two-point correlation function and the power spectral density of the nucleus point pattern to their analytical counterparts \citep{schlicke2026hdtd}, and the calibration has been validated against simulated biopsies \citep{hofmann2026}. The parameters are obtained on a normalized time scale, since one slide does not fix the absolute time scale \citep{schlicke2026hdtd,hofmann2026}. This work addresses the corresponding question for several different and interacting cell populations: what do the first- and second-order spatial statistics of one stationary snapshot reveal about the interaction matrix of a linear multi-population reaction-diffusion model?

For a single stationary covariance the antisymmetric part of a linear drift, measured against the stationary covariance, is the irreversible circulating component of the dynamics, and it is generically not determined by the stationary distribution alone \citep{godreche2019,loos2020,zhang2023}. The known ways around this bound restrict the interaction graph, use higher-order cumulants, or observe the stationary law under several experimental conditions \citep{kuegler2014,salehkaleybar2026}. A spatial snapshot offers a different path: it does not only provide a single covariance, but a one-parameter family of covariances dependent on the wavenumber, and diffusion deforms the drift across this family by $-|k|^2D$ with the unknown diffusivity matrix $D$. This means that the interaction matrix $A$ and the diffusivity matrix $D$ have to be recovered jointly.

In the following, it is proved that $A$ and $D$ are generically identifiable up to a global time rescaling from this family, for every number $K$ of interacting populations, that $\binom{K+1}{2}+1$ distinct wavenumbers suffice, that the time scale itself is not identifiable from any snapshot alone without further (external) information, and that a population-specific normalization imposed by unknown detection probabilities costs exactly the $K$ coupling scales. The following section~\ref{sec:setup} defines the model, the observable and the employed identifiability notation. Section~\ref{sec:results} contains the analytical proofs. Section~\ref{sec:application} applies the model to the example setup for tumor-immune tissue, and Section~\ref{sec:discussion} states the consequences, the relation to previous work and the limitations.

\section{Definitions and observables}\label{sec:setup}

Let $\rho(x,t) \in \mathbb{R}^K$ be the vector of densities of $K$ populations at position $x \in \mathbb{R}^d$ and time $t$. Assume the population dynamics have reached a stationary and spatially homogeneous state, then the stationary mean $\lambda := \mathbb{E}[\rho(x,t)] \in \mathbb{R}^K$ does not depend on $x$ or $t$. The fluctuation around it can be expressed as $u(x,t)=\rho(x,t)-\lambda$ and $u: \mathbb{R}^d \times \mathbb{R}_{\ge 0} \to \mathbb{R}^K$ satisfies the linear stochastic reaction-diffusion equation
\begin{equation}
\partial_t u(x,t) = Au(x,t) + D \nabla^2 u(x,t) + \xi(x,t)
\label{eq:rde}
\end{equation}
as the linearization of the population dynamics around $\lambda$, with $A \in \mathbb{R}^{K\times K}$, $D = \mathrm{diag}(D_1,\dots,D_K)$ with $D_i \geq 0$ for all $i$, and Gaussian noise
\begin{equation}
\mathbb{E} \left[\xi(x,t) \xi(x',t')^{T}\right] = \Sigma \delta(x-x') \delta(t-t'),
\qquad \Sigma = \Sigma^{T} \text{ positive definite}.
\label{eq:noise}
\end{equation}
Equation \eqref{eq:rde} is the linearization of the population dynamics around the stationary state $\lambda$. If the full dynamics read $\partial_t \rho(x,t) = F(\rho(x,t)) + D \nabla^2 \rho(x,t) + \xi(x,t)$ with $F(\lambda)=0$, then we have that the matrix $A=\partial F/\partial \rho (\lambda)$ is the Jacobian of the reaction terms evaluated at $\lambda$. The diagonal entries $A_{ii}$ are the net turnover rates for each respective population, and the off-diagonal entries $A_{ij}$ for $i \neq j$ describe the directed effect of population $j$ on population $i$. The motility of population $i$ is denoted by $D_i$, and $\Sigma$ is the noise covariance.

For each wavenumber $k$ the Fourier transform $\hat u(k,t) = \int u(x,t)e^{-ik\cdot x}dx$ satisfies the linear equation
\[
\partial_t \hat u(k,t) = M(k) \hat u(k,t) + \hat \xi(k,t), \qquad
M(k) := A - |k|^2 D \in \mathbb{R}^{K\times K},
\]
where $\hat \xi(k,t)$ is the Fourier transform of $\xi(x,t)$. Equation \eqref{eq:noise} assumes white noise. This means that $\Sigma$ does not depend on $k$ and therefore the wavenumber can only affect the model via $M(k)$.

\begin{assumption}[Stability]
\label{ass:stab}
For every $k \in \mathbb{R}^d$ the eigenvalues of the matrix $M(k)$ have negative real part (i.e., $M(k)$ is Hurwitz). Since $D_i \geq 0$ for all $i$, a sufficient condition is that the symmetric part of $A$ is negative definite.
\end{assumption}
Under Assumption~\ref{ass:stab} the fluctuation $u(x,t)$ has a stationary and unique Gaussian distribution with mean zero and spatial covariance $C(r)=\mathbb{E}[u(x)u(x+r)^{T}] \in \mathbb{R}^{K \times K}$ in the sense of generalized random fields. The latter depends on $r$, and is independent of $x$. The structure factor $S(k) \in \mathbb{C}^{K\times K}$ is the Fourier transform of $C(r)$, i.e. $S(k) = \int C(r) e^{-ik\cdot r} dr \in \mathbb{C}^{K\times K}$ and its entries $S_{ij}(k)$ are the cross spectra of populations $i$ and $j$ at wavenumber $k$, and as such, $S_{ii}(k)$ is the power spectrum of $i$.
Equivalently, $S(k)$ is the stationary covariance of $\hat u(k,\cdot)$. It therefore solves the algebraic Lyapunov equation
\begin{equation}
M(k) S(k) + S(k) M(k)^{T} + \Sigma = 0.
\label{eq:lyap}
\end{equation}
For a marked point pattern with intensities $\lambda_i$ and cross pair correlation functions $g_{ij}(r)$, we have the multitype version of the spectrum \citep{bartlett1964}
\begin{equation}
S_{ij}(k) = \lambda_i\delta_{ij} + \lambda_i\lambda_j\int \bigl(g_{ij}(r)-1\bigr)e^{-ik\cdot r} dr.
\label{eq:bartlett}
\end{equation}
The first term $\lambda_i\delta_{ij}$ does not depend on $k$, and it is the spectrum that a Poisson pattern with the same intensities would have. The second term supports the spatial correlations.

The data of a snapshot are the stationary mean $\lambda$ and the structure factor $S(k)$. In the linearized model $\lambda$ is independent of $(A,D,\Sigma)$, which means that identifiability of $(A,D,\Sigma)$ only depends on $S(\cdot)$. Throughout, $\kappa=|k|^2$, and $A_s$ and $A_a$ denote the symmetric and antisymmetric parts of $A$. We introduce the following definitions:

\begin{definition}[Parameters, observation map, inverse image]
\label{def:setup}
$\Theta$ denotes the set of triples $\theta=(A,D,\Sigma)$ satisfying Assumption~\ref{ass:stab} with $D\geq0$ diagonal and $\Sigma$ positive definite. We have $\Theta \subset \mathbb{R}^{n_U}$ with $n_U=K^2+K+\binom{K+1}{2}$. The observation map is denoted as $\Phi:\Theta\to\{S(\kappa):\kappa\ge0\}$, $\theta\mapsto S_\theta(\cdot)$, with $S_\theta(\kappa)$ the solution of equation \eqref{eq:lyap}. The inverse image of $\theta_0 \in \Theta$ is $\Phi^{-1}(\Phi(\theta_0))\cap\Theta$, namely the set of parameters in $\Theta$ with the same structure factor as $\theta_0$.
\end{definition}

\begin{definition}[Identifiability]
\label{def:ident}
Let $\theta_0\in\Theta$. The triple $(A,D,\Sigma)$ is globally identifiable at $\theta_0$ up to the time scale if the inverse image of $\theta_0$ is the half-line $\{s\theta_0: s>0\}$, and globally identifiable if the inverse image is $\{\theta_0\}$. It is locally identifiable at $\theta_0$ up to the time scale if $\ker d\Phi|_{\theta_0}=\mathrm{span}\{\theta_0\}$, and locally identifiable if $\ker d\Phi|_{\theta_0}=\{0\}$. For fixed $(D,\Sigma)$ the matrix $A$ is called identifiable if from $\Phi(A',D,\Sigma)=\Phi(A,D,\Sigma)$ and $(A',D,\Sigma)\in\Theta$ follows $A'=A$.
\end{definition}

\begin{definition}[Generic]
\label{def:generic}
A feature holds for a generic $\theta \in \Theta$ if there is a nonzero polynomial $q$ on $\mathbb{R}^{n_U}$ such that the feature holds for any $\theta \in \Theta$ with $q(\theta) \neq 0$. The zero set of $q$ has Lebesgue measure zero. The feature therefore holds on a dense subset of full measure of the interior $\Theta^o = \{\theta \in \Theta : D_i>0 \ \forall i\}$.
\end{definition}

\section{Results}\label{sec:results}

\begin{lemma}[Structure of $S(k)$]
\label{lem:sym}
With Assumption~\ref{ass:stab} and for every $k$ the solution $S(k)$ of \eqref{eq:lyap} is a real symmetric positive-definite matrix that only depends on $|k|^2$.
\end{lemma}
\begin{proof}
Let $S(k)=S_R(k)+iS_I(k)$ with $S_R(k),S_I(k)$ real valued and $i$ the imaginary unit. Substitute into \eqref{eq:lyap}, then separate real and imaginary parts to receive
\[
M S_R + S_R M^{T} + \Sigma = 0, \qquad M S_I + S_I M^{T} = 0.
\]
The second equation is a homogeneous Lyapunov equation. Since $M$ is Hurwitz, the only solution is $S_I(k)=0$. We receive $S(k)=S_R(k)$, i.e., $S(k)$ is real. Transposing \eqref{eq:lyap} with $\Sigma=\Sigma^{T}$ shows that $S^{T}(k)$ also fulfills the equation, and by uniqueness we receive $S^{T}(k)=S(k)$. From $\Sigma$ positive definite and $M$ Hurwitz we have $S(k) = \int_0^\infty e^{Mt} \Sigma e^{M^{T}t} dt$, and thus $S$ is positive definite. $M(k)$ depends on $|k|^2$, and thus also $S(k)$.
\end{proof}

Thus, within the model class, $S(k)$ does not have an imaginary or an anisotropic part. Recoverable information is fully contained in the real symmetric family $\{S(\kappa)\}$, and directionality is identified via the dependence of this family on $\kappa$.

\begin{proposition}[Time scale invariance]
\label{prop:scale}
For any $c>0$, we have that $\Phi(A,D,\Sigma) = \Phi(cA,cD,c\Sigma)$, and $\lambda$ remains identical. Consequently, if $\Sigma$ is a free parameter, $(A,D,\Sigma)$ is identifiable only up to the one-parameter scaling, and the absolute time scale remains not identifiable.
\end{proposition}

\begin{proof}
The stationary state $\lambda$ around which \eqref{eq:rde} is linearized is independent of $(A,D,\Sigma)$. In the affine model $\partial_t \rho = A \rho + b + D \nabla^2 \rho + \xi$ we have $\lambda=-A^{-1}b=-(cA)^{-1}(cb)$, thus $\lambda$ is unchanged under rescaled $b$ with factor $c$. For the mapping $(A,D,\Sigma)\mapsto(cA,cD,c\Sigma)$ we receive $M(k)\mapsto cM(k)$, so the left-hand side of \eqref{eq:lyap} maps to $c\left(M(k)S(k)+S(k)M(k)^{T}+\Sigma\right)$. The transformed equation is also solved by $S(k)$, since $c>0$. Equivalently, the transformation is a time rescaling $t\mapsto t/c$.
\end{proof}

\begin{remark}
\label{rem:scaleknown}
If $\Sigma$ is known, the rescaling is not available. In that case, the scale is identifiable (see Theorem~\ref{thm:global}). Non-identifiability is therefore a consequence of unknown $\Sigma$. The single-population calibration in \cite{schlicke2026hdtd} addressed this same invariance through a normalized time scale. The dispersion length $\sqrt{D/\gamma}$, with $\gamma$ the net proliferation rate, is invariant under the time rescaling and can be recovered without knowledge about the time scale in simulated biopsies \citep{hofmann2026}.
\end{remark}

\begin{proposition}[Invariance of $S(k)$ under perturbed $A$]
\label{prop:inv}
Let $(D,\Sigma)$ be fixed. For $A+J-|k|^2D$ Hurwitz for all $k$, the perturbation of the interaction matrix $A\mapsto A+J$ leaves every structure factor $S(k)$ unchanged if and only if
\begin{equation}
J S(k) + S(k) J^{T} = 0 \qquad \text{for all } k.
\label{eq:inv}
\end{equation}
In particular, a purely symmetric $J$ satisfies \eqref{eq:inv} if and only if $J=0$, while a purely antisymmetric $J$ satisfies it if and only if $[J,S(k)] = 0$ for all $k$.
\end{proposition}

\begin{proof}
Replacing $A \mapsto A+J$ yields $M(k) \mapsto M(k)+J$. If $S(k)$ solves \eqref{eq:lyap} with the same $\Sigma$ before and after the perturbation, subtracting gives $JS(k)+S(k)J^{T}=0$. Conversely, if $S(k)$ solves \eqref{eq:inv}, then $S(k)$ solves the perturbed Lyapunov equation and, since it is the unique solution, it also is the new structure factor. We have $S(k)=S(k)^{T}$ (cf. proof of Lemma \ref{lem:sym}) and $(J S(k))^{T} = S(k) J^{T}$, so \eqref{eq:inv} yields that $JS(k)$ is antisymmetric. Let $J$ be symmetric, then $J S(k) + S(k) J^{T} = J S(k) + S(k) J$, and the linear map $X \mapsto X S(k) + S(k) X$ on symmetric matrices $X$ has the eigenvalues $\mu_i+\mu_j>0$, where $\mu_1 , ... , \mu_K > 0$ are the eigenvalues of $S(k)$. This means the map is injective, and $J=0$. Let $J$ be antisymmetric, then $JS(k) + S(k) J^{T} = JS(k) - S(k) J = [J,S(k)]$. For a general $J=J_s + J_a$ the two parts couple, since \eqref{eq:inv} reads $J_s S(k) + S(k) J_s+[J_a,S(k)]=0$.
\end{proof}

\begin{corollary}[Identifiability of $A$ for known $(D,\Sigma)$]
\label{cor:id}
Let $(D,\Sigma)$ be fixed. The interaction matrix $A$ is identifiable from $\{S(k)\}$ if and only if $JS(k)+S(k)J^{T}=0$ for all $k$ is only fulfilled by $J=0$. In that case both the symmetric part $A_s$ and the antisymmetric part $A_a$ are determined.
\end{corollary}

\begin{proof}
Two matrices $A,A'$ produce the same $\{S(k)\}$ for a fixed $(D,\Sigma)$ if and only if $J:=A'-A$ satisfies equation \eqref{eq:inv}.

If $J$ solves \eqref{eq:inv} we have $(A+J,D,\Sigma) \in \Theta$ and the positive definite $S(k)$ solves the Lyapunov equation of $A+J-\kappa D$ with positive definite $\Sigma$. By Lyapunov's theorem, we receive that $A+J-\kappa D$ is Hurwitz for every $k$. Thus, $A$ is identifiable if and only if $J=0$ is the unique solution.
\end{proof}

In particular, a nonzero antisymmetric matrix $J$ that commutes with every $S(k)$ satisfies \eqref{eq:inv}. Then, $A_a$ is not identifiable. Corollary~\ref{cor:id} assumes that $(D,\Sigma)$ is known. Combined with the scaling of Proposition~\ref{prop:scale} it does not by itself guarantee that $(A,D,\Sigma)$ is jointly identifiable.

We now turn to the joint identifiability of $(A,D,\Sigma)$ and first describe the kernel of the differential of $\Phi$. Let $\theta_0=(A,D,\Sigma)\in\Theta$ and let $\delta\theta=(\delta A,\delta D,\delta\Sigma)$ denote a parameter perturbation, with $\delta A\in\mathbb{R}^{K\times K}$, $\delta D$ diagonal and $\delta\Sigma$ symmetric, such that $\theta_0+\varepsilon \delta\theta$ is again a parameter triple for $\varepsilon$ small. The differential $d\Phi|_{\theta_0}$ maps $\delta\theta$ to the first-order change $\delta S(\kappa)$ of the structure factor. Differentiating \eqref{eq:lyap} with respect to $\varepsilon$ close to $0$ gives
\begin{equation}
\begin{aligned}
&M(\kappa) \delta S(\kappa)+\delta S(\kappa) M(\kappa)^{T}+E(\kappa)=0,\\
&E(\kappa):=\delta A S(\kappa)+S(\kappa) \delta A^{T}-\kappa\bigl(\delta D S(\kappa)+S(\kappa) \delta D\bigr)+\delta\Sigma,
\label{eq:E}
\end{aligned}
\end{equation}
with $S=S_{\theta_0}$. Here, $E(\kappa)$ is the change of the left-hand side of \eqref{eq:lyap} as a consequence of the perturbation for fixed $S$. Since $X\mapsto M(\kappa)X+XM(\kappa)^{T}$ is injective for $M(\kappa)$ Hurwitz, $\delta S(\kappa)=0$ if and only if $E(\kappa)=0$. Therefore
\begin{equation}
\ker d\Phi|_{\theta_0} = \bigl\{(\delta A,\delta D,\delta\Sigma) : E(\kappa) = 0 \ \text{ for all } \kappa \ge 0 \bigr\}.
\label{eq:kernel}
\end{equation}
The scaling direction $\delta\theta=\theta_0$ is in the kernel: for $(\delta A,\delta D,\delta\Sigma)=(A,D,\Sigma)$ the function $E(\kappa)$ equals the left-hand side of \eqref{eq:lyap}. Thus, $\dim\ker d\Phi\ge1$ on $\Theta$. The following lemma shows that the condition in Equation \eqref{eq:kernel} reduces to finitely many linear equations in $(\delta A,\delta D,\delta\Sigma)$.

\begin{lemma}[Finite determination]
\label{lem:finite}
Let $\theta_0\in\Theta$, let $N_\ast:=\binom{K+1}{2}$, and let $p(\kappa)$ be the determinant of the linear map $L_\kappa: X\mapsto M(\kappa)X+XM(\kappa)^{T}$ on $\mathrm{Sym}(K)$ of symmetric $K\times K$ matrices. For every perturbation $(\delta A,\delta D,\delta\Sigma)$ the product $p(\kappa)E(\kappa)$ is a polynomial in $\kappa$ of degree less or equal $N_\ast$, and the following three statements are equivalent: (i) $E(\kappa)=0$ for all $\kappa\ge0$; (ii) $E(\kappa_j)=0$ at $N_\ast+1$ pairwise distinct values $\kappa_0,\dots,\kappa_{N_\ast}\ge0$; (iii) $E_0=\dots=E_{N_\ast}=0$, where $E(\kappa)=\sum_{n\ge0}E_n\kappa^n$ is the power series of $E$ at $\kappa=0$. The coefficients are $E_0=\delta A S_0+S_0 \delta A^{T}+\delta\Sigma$ and $E_n=\delta A S_n+S_n \delta A^{T}-(\delta D S_{n-1}+S_{n-1} \delta D)$ for $n\ge1$, where $S(\kappa)=\sum_{n\ge0}S_n\kappa^n$ and the matrices $S_n$ are recursively available by $AS_0+S_0A^{T}=-\Sigma$ and $AS_n+S_nA^{T}=DS_{n-1}+S_{n-1}D$ for $n\ge1$.
\end{lemma}

\begin{proof}
In a basis of $\mathrm{Sym}(K)$ that has dimension $N_\ast$, the map $L_\kappa$ is an $N_\ast\times N_\ast$ matrix whose entries are affine functions of $\kappa$, since $M(\kappa)=A-\kappa D$. Hence $p(\kappa)=\det L_\kappa$ is a polynomial of degree less or equal $N_\ast$, and the entries of $\mathrm{adj}(L_\kappa)$ are polynomials of degree less or equal $N_\ast-1$. Both bounds use only the dimension of $\mathrm{Sym}(K)$ and they also hold if some $D_i=0$. By Assumption~\ref{ass:stab}, $L_\kappa$ is invertible for every $\kappa\ge0$, so $p(\kappa)\ne0$ there, and Cramer's rule gives $S(\kappa)=\mathcal{N}(\kappa)/p(\kappa)$, where $\mathcal{N}(\kappa)\in\mathrm{Sym}(K)$ is obtained by applying $\mathrm{adj}(L_\kappa)$ to $-\Sigma$ and has entries that are polynomials of degree less or equal $N_\ast-1$. Multiplying $E$ by $p$ gives
\[
p(\kappa)E(\kappa)=\delta A \mathcal{N}(\kappa)+\mathcal{N}(\kappa) \delta A^{T}-\kappa\bigl(\delta D \mathcal{N}(\kappa)+\mathcal{N}(\kappa) \delta D\bigr)+p(\kappa) \delta\Sigma,
\]
and every term on the right is a polynomial of degree less or equal $N_\ast$. This proves the first claim and in particular, $E$ is a rational function of $\kappa$ without poles on $[0,\infty)$.
(i) implies (ii) trivially. If (ii) holds, the polynomial $pE$, of degree less or equal $N_\ast$, vanishes at $N_\ast+1$ distinct points, so it is the zero polynomial, and $E=0$ wherever $p\ne0$, in particular on $[0,\infty)$. This yields (i). For (iii), note that $E$ is analytic in a neighborhood of $\kappa=0$ with a power series of positive radius of convergence, since $p(0)\ne0$. If (i) holds, then $E$ and all coefficients $E_n$ vanish. Conversely, if $E_0=\dots=E_{N_\ast}=0$, then $E(\kappa)=O(\kappa^{N_\ast+1})$ near $\kappa=0$, thus $p(\kappa)E(\kappa)=O(\kappa^{N_\ast+1})$. A polynomial of degree less or equal $N_\ast$ with a zero of order $N_\ast+1$ at the origin is zero, so $E\equiv0$. The formulas for $E_n$ and $S_n$ follow by inserting $S(\kappa)=\sum_n S_n\kappa^n$ into the definition of $E$ and into \eqref{eq:lyap} and comparing the coefficients of $\kappa^n$.
\end{proof}

\begin{lemma}[Genericity of minimal kernel]
\label{lem:generic}
There is a polynomial $q$ on $\mathbb{R}^{n_U}$ such that for $\theta\in\Theta$
\[
\dim\ker d\Phi|_{\theta}=1 \quad\text{if and only if}\quad q(\theta)\neq0 .
\]
Consequently, if $\dim\ker d\Phi|_{\theta}=1$ at one point of $\Theta$, then $q$ is not the zero polynomial and $\dim\ker d\Phi|_{\theta}=1$ holds for generic $\theta\in\Theta$ in the sense of Definition~\ref{def:generic}.
\end{lemma}

\begin{proof}
By \eqref{eq:kernel} and Lemma~\ref{lem:finite}(iii), $\ker d\Phi|_{\theta}$ is the solution space of the finite linear system $E_0=\dots=E_{N_\ast}=0$ in the unknown $(\delta A,\delta D,\delta\Sigma)\in\mathbb{R}^{n_U}$. Let $C(\theta)$ be its coefficient matrix, such that $\dim\ker d\Phi|_{\theta}=n_U-\mathrm{rank} C(\theta)$. The entries of $C(\theta)$ are linear in the entries of $S_0,\dots,S_{N_\ast}$, and these are received from the recursion of Lemma~\ref{lem:finite} through solving $N_\ast+1$ Lyapunov equations with $L_0: X\mapsto AX+XA^{T}$. By Cramer's rule these are rational functions of $\theta$ and their denominators are powers of $\det L_0$. Since $A=M(0)$ is Hurwitz on $\Theta$ we have $\det L_0\neq0$, and multiplying $C(\theta)$ by a sufficiently high power of $\det L_0$ yields a matrix $\tilde C(\theta)$ with polynomial entries and $\mathrm{rank} \tilde C(\theta)=\mathrm{rank} C(\theta)$ for all $\theta\in\Theta$. The scaling direction $\theta$ solves the system, and $\mathrm{rank} C(\theta)\le n_U-1$ on $\Theta$, with equality exactly when $\dim\ker d\Phi|_{\theta}=1$. Let $q(\theta)$ be the sum of the squares of all minors of size $n_U-1$ of $\tilde C(\theta)$. Then $q$ is a polynomial on $\mathbb{R}^{n_U}$, and $q(\theta)\neq0$ if and only if at least one of these minors is nonzero, i.e., if and only if $\mathrm{rank} \tilde C(\theta)\ge n_U-1$. For $\theta\in\Theta$ this is exactly the case when $\dim\ker d\Phi|_{\theta}=1$. If this happens at one point of $\Theta$, then $q$ cannot be the zero polynomial, and Definition~\ref{def:generic} applies.
\end{proof}

\begin{theorem}[Generic identifiability up to the time scale]
\label{thm:main}
At a generic $(A,D,\Sigma)\in\Theta$ and for every number $K \ge 1$ of populations, the differential of the observation map $\Phi$ has a one-dimensional kernel that is spanned by the scaling direction of Proposition~\ref{prop:scale}. Thus, $(A,D,\Sigma)$ is locally identifiable from $\{S(\kappa)\}$ up to that rescaling. If $\Sigma$ is known in addition, the kernel is trivial and the scale becomes identifiable.
\end{theorem}

\begin{proof}
By Lemma~\ref{lem:generic} it suffices to find one point of $\Theta$ with $\dim\ker d\Phi=1$. We take a point at which $A$ is diagonal, so that the Lyapunov equation decouples into scalar equations. Let $A=\mathrm{diag}(a_1,\dots,a_K)$ and $D=\mathrm{diag}(d_1,\dots,d_K)$ with $a_i<0$ and $d_i>0$, chosen such that
\begin{equation}
\text{the } \tbinom{K+1}{2} \text{ numbers } \kappa_{ij}:=\frac{a_i+a_j}{d_i+d_j}, \ i\le j,
\ \text{are pairwise distinct,}
\label{eq:polesep}
\end{equation}
and let $\Sigma$ be symmetric positive definite with nonzero entries. Such points exist for every $K$: the choice $d_i=1$, $a_i=-2^{i}$ satisfies \eqref{eq:polesep}, since $\kappa_{ij}=-(2^{i}+2^{j})/2$ and the binary expansion of $2^{i}+2^{j}$ determines the pair $\{i,j\}$ (for $i=j$ the number is $2^{i+1}$). $\Sigma=(1-\varepsilon)I+\varepsilon \mathbf{1}\mathbf{1}^{T}$ with $0<\varepsilon<1$ and the vector of ones $\mathbf{1}$ has all entries nonzero and the positive eigenvalues $1+\varepsilon(K-1)$ and $1-\varepsilon$ (for $K=1$ read $\Sigma=1$). Condition \eqref{eq:polesep} is stronger than requiring distinct ratios $a_i/d_i$: for example for $a=(-1,-2,-3)$ and $d=(1,1,1)$ the ratios are distinct, but $\kappa_{13}=\kappa_{22}=-2$. $A_a=0$ is not relevant for our purpose, since our claim is about the rank of $d\Phi$.

At that point $M(\kappa)$ is diagonal and \eqref{eq:lyap} decouples for each entry:
\[
S_{ij}(\kappa) = \frac{-\Sigma_{ij}}{\mathcal{D}_{ij}(\kappa)}, \qquad
\mathcal{D}_{ij}(\kappa) = (a_i+a_j) - \kappa(d_i+d_j).
\]
Each entry has a single simple pole at $\kappa_{ij}<0$, and these poles are pairwise distinct by \eqref{eq:polesep}. Let $(\delta A,\delta D,\delta\Sigma)\in\ker d\Phi$. By Lemma~\ref{lem:finite}, $E(\kappa)=0$ for $\kappa\ge0$ implies that $E$ vanishes identically as a rational function, and in particular every residue of every entry of $E$ vanishes. Consider now the diagonal entry $E_{ii}(\kappa)$. Among its terms, only $\delta A_{il}S_{li}(\kappa)+S_{il}(\kappa)\delta A_{il}=2\delta A_{il}S_{il}(\kappa)$ (using $S=S^{T}$) has a pole at $\kappa_{li}$ for $l\ne i$, and
\[
\operatorname*{Res}_{\kappa=\kappa_{li}} E_{ii}
= \frac{2\delta A_{il}\Sigma_{li}}{d_l+d_i}.
\]
Since $\Sigma_{li}\ne0$, this yields $\delta A_{il}=0$. Letting $l\ne i$ and $i$ range over all indices shows that $\delta A$ is diagonal. This step uses only that for fixed $i$ the $K$ poles $\kappa_{1i},\dots,\kappa_{Ki}$ are distinct.

With $\delta A$ diagonal, the $(i,j)$ entry of $E$ for $i\ne j$ reads $\big((\delta A_{ii}+\delta A_{jj})-\kappa(\delta D_i+\delta D_j)\big)S_{ij}(\kappa)+\delta\Sigma_{ij}=0$. Multiply by $\mathcal{D}_{ij}(\kappa)$ to receive an identity between polynomials of degree one in $\kappa$, and compare coefficients for
\[
\delta\Sigma_{ij} = \Sigma_{ij}\frac{\delta D_i+\delta D_j}{d_i+d_j}, \qquad
\delta A_{ii}+\delta A_{jj} = (a_i+a_j)\frac{\delta D_i+\delta D_j}{d_i+d_j}.
\]
The $(i,i)$ entry gives in the same way $\delta A_{ii}=a_i\delta D_i/d_i$ and $\delta\Sigma_{ii}=\Sigma_{ii}\delta D_i/d_i$. Writing $f_i := \delta D_i/d_i$, the off-diagonal relation becomes $a_i f_i + a_j f_j = (a_i+a_j)(d_i f_i + d_j f_j)/(d_i+d_j)$, which simplifies to $(f_i-f_j)(a_i d_j - a_j d_i)=0$. By \eqref{eq:polesep} the values $\kappa_{ii}=a_i/d_i$ are pairwise distinct, so $a_id_j\ne a_jd_i$ and $f_i=f_j=:c$ for all $i,j$. Then $\delta D=cD$, $\delta A=cA$ and $\delta\Sigma=c\Sigma$, i.e., $\ker d\Phi=\mathrm{span}\{(A,D,\Sigma)\}$ at this point, and by Lemma~\ref{lem:generic} at generic points of $\Theta$. If $\Sigma$ is known, $\delta\Sigma=0$ forces $c=0$, which yields scale identifiability.
\end{proof}

The next lemma extends the local statement to a global statement. It uses that for fixed $S$, the Lyapunov equation is linear in the parameters.

\begin{lemma}[Inverse image is affine]
\label{lem:inverseimage}
For $\theta_0\in\Theta$,
\[
\Phi(\theta) = \Phi(\theta_0)
\quad \text{ if and only if } \quad
\theta-\theta_0 \in \ker d\Phi|_{\theta_0},
\qquad \theta\in\Theta .
\]
If $\ker d\Phi|_{\theta_0}=\mathrm{span}\{\theta_0\}$, then the inverse image of $\theta_0$ is exactly the half-line: $\Phi^{-1}(\Phi(\theta_0)) = \{s\theta_0 : s>0\}$.
\end{lemma}

\begin{proof}
Let $S:=\Phi(\theta_0)$. If $\Phi(\theta)=S$, then both $\theta$ and $\theta_0$ satisfy \eqref{eq:lyap} with the same $S(\kappa)$ for every $\kappa\ge0$. Subtracting the two results in relation \eqref{eq:kernel} of $\ker d\Phi|_{\theta_0}$ for $\theta-\theta_0$. Conversely, if $\theta-\theta_0\in\ker d\Phi|_{\theta_0}$, adding \eqref{eq:lyap} for $\theta_0$ shows that $S(\kappa)$ solves the Lyapunov equation for $\theta$. Since $M_\theta(\kappa)$ is Hurwitz, this solution is unique, and therefore $\Phi(\theta)=S$. Finally, $\theta=s\theta_0$ lies in $\Theta$ if and only if $s>0$. A negative $s<0$ reverses the spectrum of $A$ and destroys stability, and $s=0$ would violate that $\Sigma$ is positive definite.
\end{proof}

\begin{theorem}[Global identifiability up to the time scale]
\label{thm:global}
At a generic $\theta_0\in\Theta$ the inverse image $\Phi^{-1}(\Phi(\theta_0))\cap\Theta$ is exactly $\{s\theta_0:s>0\}$. Therefore, $(A,D,\Sigma)$ is globally identifiable up to a positive time rescaling. In particular, the signs and the ratios of the entries of the directed part $A_a$ are identifiable. If $\Sigma$ is known additionally, it is even globally identifiable.
\end{theorem}

\begin{proof}
Combine Theorem~\ref{thm:main} with Lemma~\ref{lem:inverseimage}.
\end{proof}

\begin{corollary}[Finitely many wavenumbers]
\label{cor:shells}
Let $\theta,\theta_0\in\Theta$ and let $\kappa_0,\dots,\kappa_{N_\ast}\ge0$ be $N_\ast+1=\binom{K+1}{2}+1$ pairwise distinct wavenumbers. If $S_\theta(\kappa_j)=S_{\theta_0}(\kappa_j)$ for all $j$, then we receive $S_\theta=S_{\theta_0}$ on $[0,\infty)$. For a generic $\theta_0$ equal structure factors at $\binom{K+1}{2}+1$ distinct wavenumbers implies $\theta=s\theta_0$ for $s>0$.
\end{corollary}

\begin{proof}
Let $E$ be the function as in equation \eqref{eq:E} with $S=S_{\theta_0}$ and $(\delta A,\delta D,\delta\Sigma)=\theta-\theta_0$. For each $\kappa_j$, subtracting the Lyapunov equations of $\theta_0$ and $\theta$ (which share the matrix $S_{\theta_0}(\kappa_j)$) gives $E(\kappa_j)=0$. By Lemma~\ref{lem:finite}, $E=0$ on $[0,\infty)$, thus $\theta-\theta_0\in\ker d\Phi|_{\theta_0}$, and from Lemma~\ref{lem:inverseimage} we receive $\Phi(\theta)=\Phi(\theta_0)$. For the second claim, let $\theta_0$ be generic. Then $\theta$ lies in the inverse image of $\theta_0$, which is the half-line $\{s\theta_0:s>0\}$ by Theorem~\ref{thm:global}.
\end{proof}

\begin{remark}[Where identifiability gets lost]
\label{rem:degen}
Theorem~\ref{thm:global} holds at generic points, but there is a set of exceptions. The first is full isotropy. For $A_s=\alpha I$, $D=dI$ and $\Sigma=\sigma I$ the family is $S(\kappa)=s(\kappa)I$ with $s(\kappa)=\sigma/(2(\kappa d-\alpha))$, for any $A_a$, and $E=0$ reduces to $\delta\Sigma=(\sigma/d)\delta D$ and $\delta A_s=(\alpha/d)\delta D$ with $\delta D$ diagonal and free $\delta A_a$. The kernel has dimension $K(K+1)/2$ and $A_a$ is lost. Equal motilities $D_i$ are not alone responsible for this behavior. On the set $D\propto I$, the kernel stays one-dimensional at generic points, since $q$ of Lemma~\ref{lem:generic} does not vanish identically there. The reference point of the proof of Theorem~\ref{thm:main} has $D=I$. The second we focus on is decoupling. If $A$, $D$ and $\Sigma$ split simultaneously into $m$ non-interacting blocks, then each block has its own time scale and $\dim\ker d\Phi\ge m$. For example, $K=2$ with $A=\mathrm{diag}(-3,-5)$, $D=\mathrm{diag}(5,2)$, $\Sigma=\mathrm{diag}(7,3)$ already has kernel dimension $2$: the off-diagonal entry of $E$ forces $\delta A_{12}=\delta A_{21}=\delta\Sigma_{12}=0$, and each block is a $K=1$ problem with the kernel of Theorem~\ref{thm:main}. A split into $r$ and $K-r$ populations results in $3r(K-r)$ conditions, so this set has positive codimension and Theorem~\ref{thm:main} is unaffected. On the decoupled set the half-line of $A_a$ can remain identifiable: for $A$ with diagonal blocks $\bigl(\begin{smallmatrix}-2&1\\-1&-3\end{smallmatrix}\bigr)$ and $-1$, $D=\mathrm{diag}(1,2,\tfrac32)$ and $\Sigma$ with diagonal blocks $\bigl(\begin{smallmatrix}1&1/2\\1/2&1\end{smallmatrix}\bigr)$ and $1$, an exact computation gives a kernel spanned by the two block scalings, which change $A_a$ only by a factor. The third is the reversible case: for diagonal $\Sigma$ with $\Sigma ^{-1} A$ symmetric, including the case $A=A^{T}$ and $\Sigma=\sigma I$, the matrix $\Sigma^{-1} M(k)$ is symmetric, so $S(\kappa) = -1/2 \ M(\kappa)^{-1} \Sigma$ is symmetric and a solution of \eqref{eq:lyap} since $M(\kappa)S(\kappa)=S(\kappa)M(\kappa)^T=-1/2 \ \Sigma$.
For every diagonal $G$ with positive entries the triple $(GA,GD,G\Sigma)$ has the same family, since 

$(GM(\kappa))^{-1} G \Sigma = M(\kappa)^{-1} \Sigma$. Then the kernel contains the $K$-dimensional space of directions $(FA,FD,F\Sigma)$ with $F$ diagonal, and $GA$ is in general not symmetric. This means that a symmetric operator with isotropic noise cannot be distinguished from a directed operator with the same sign pattern. This set has codimension $K^2-K$, and with $A=A^{T}$ and generic diagonal $\Sigma$ the kernel is one-dimensional in exact computations for $K \leq 5$.
\end{remark}

The observation from a tissue section does not have to identify the raw family $\{S(\kappa)\}$. If the cells of population $i$ are detected and typed with an unknown probability $p_i$, then independent thinning preserves the pair correlation functions $g_{ij}$ and scales the second term of Equation~\ref{eq:bartlett} by $p_i p_j$, so each population is observed only up to a factor. The following corollary shows that this removes the coupling scales.

\begin{corollary}[Normalized observation]
\label{cor:norm}
Let the field of population $i$ be only observed up to an unknown scalar $g_i>0$, such that the observation is $\{GS(\kappa)G\}$ with $G=\mathrm{diag}(g_1,\dots,g_K)$ unknown. Equivalently, let the family be normalized per population, $S_{ij}(\kappa)/(\sigma_i\sigma_j)$ with $\sigma_i$ a fixed amplitude for population $i$. At a generic $\theta_0=(A,D,\Sigma)\in\Theta$ the set of admissible triples with the same normalized observation is therefore
\[
\bigl\{(sGAG^{-1}, sD, sG\Sigma G): s>0, G=\mathrm{diag}(g_1,\dots,g_K) \text{ with $g_i>0$ for all $i$}\bigr\},
\]
a manifold of dimension $K+1$. The signs and zero pattern of the entries of $A$ are identifiable, the products of entries along closed paths and the diffusivities are identifiable up to the time scale, and the $K$ factors $g_i$ are not identifiable from the second order.
\end{corollary}

\begin{proof}
For diagonal positive $G$ and $\theta=(A,D,\Sigma)\in\Theta$ define $\theta^{G}:=(GAG^{-1},D,G\Sigma G)$. Since $G$ commutes with the diagonal $D$, $M_{\theta^{G}}(k)=GM_{\theta}(k)G^{-1}$, and multiplying \eqref{eq:lyap} for $\theta$ from the left and from the right by $G$ shows that $GS_{\theta}(\kappa)G$ solves the Lyapunov equation of $\theta^{G}$. By uniqueness we receive $\Phi(\theta^{G})=G\Phi(\theta)G$. Moreover $\theta^{G}\in\Theta$, because $GAG^{-1}$ has the eigenvalues of $A$ and $G\Sigma G$ is positive definite, and both $\theta\mapsto\theta^{G}$ and $S\mapsto GSG$ are linear isomorphisms, such that $\dim \ker d\Phi|_{\theta^{G}} = \dim \ker d\Phi|_\theta$ and $\theta^G$ is generic whenever $\theta$ is. Now let $\theta,\theta'\in\Theta$ have the identical normalized observation, that is, $\Phi(\theta')=G\Phi(\theta)G$ for a $G$ diagonal matrix with positive entries. Then $\Phi(\theta')=\Phi(\theta^{G})$ with $\theta^{G}$ generic, and from Theorem~\ref{thm:global} we have $\theta'=s\theta^{G}$ for $s>0$.
Conversely, every triple $(sGAG^{-1},sD,sG\Sigma G)$ has the normalized observation of $\theta$ by $\Phi(\theta^{G})=G\Phi(\theta)G$ and Proposition~\ref{prop:scale}. The parameters $(s,g_1,\dots,g_K)$ are determined by the triple, so the set is a manifold of dimension $K+1$. The invariants follow from $(GAG^{-1})_{ij}=g_iA_{ij}/g_j$ with $g_i/g_j>0$. For products of entries along a closed path these factors cancel. The individual factors $g_i$ themselves do not appear in the normalized observation and are therefore not identifiable from it.
\end{proof}

\section{Application to tumor-immune tissue}\label{sec:application}

Let \eqref{eq:rde} hold with $K$ cell types and $\rho_i(x,t)$ the local density of cell type $i$ (e.g. tumor, cytotoxic effector, regulatory or suppressor cell, macrophage). The diagonal entries of $A$ collect the net per-type turnover (tumor proliferation, effector expansion or exhaustion), the off-diagonal entries the directed interactions, for example killing ($A_{\mathrm{tumor},\mathrm{effector}}<0$), recruitment ($A_{\mathrm{effector},\mathrm{tumor}}>0$) and suppression ($A_{\mathrm{effector},\mathrm{regulatory}}<0$). $A$ is non-symmetric since these effects are one-directional. $D$ collects the cell-type-dependent, generally distinct cell motilities, and $\Sigma$ the demographic and interaction noise. A single multiplexed section (for example IMC, CODEX or Xenium) yields the composition $\lambda$ and the matrix of cross-cell-type pair correlations $g_{ij}(r)$, that is, the structure factor from Section~\ref{sec:setup}. For a single population we have computed the corresponding pair correlation function and spectrum from nucleus centroids with polygonal edge correction and a Hankel transform \citep{schlicke2026hdtd}, and the multi-population observation is the matrix-valued analog. In practice, the observation is obtained from a tissue section in several steps: cells are segmented and assigned to cell types, the cross-type pair correlation functions $g_{ij}(r)$ are estimated with edge correction and are transformed to $S_{ij}(\kappa)$ by \eqref{eq:bartlett} with a Hankel transform, and the Poisson term $\lambda_i\delta_{ij}$ is subtracted. Then, $(A,D,\Sigma)$, with $\Sigma$ the wavenumber-independent effective covariance of the Poisson-corrected family, is fitted to it on the resolved interval up to the time scale. By Corollary~\ref{cor:shells}, $\binom{K+1}{2}+1$ distinct wavenumbers are sufficient. The theorems state what this fit can determine in principle. The estimator, its error model and the requirements for acquisition are future work.

The diagonal of $A$ can, for example, be read against a proliferation marker. Ki67 estimates the tumor proliferative fraction, but it does not by itself identify the external clock that the time scale requires. It is degraded continuously after cell-cycle exit, so its level is a function of the time since the last division and the decay requires calibration \citep{miller2018}. The diagonal entry of the effector population is its net turnover which e.g. includes the progression from progenitor to terminally exhausted T cells. The two ends of this trajectory are marked by TCF7 together with SLAMF6 (progenitor) and TIM-3 (terminal) \citep{miller2019}, and TOX regulates the exhaustion program \citep{khan2019,scott2019}.

In general, the diffusivities differ across cell types. The tail of an autospectrum alone does not identify $D_i$ but the ratio $D_i/\Sigma_{ii}$, since $S_{ii}(\kappa) \to \Sigma_{ii}/(2\kappa D_i)$ for large $\kappa$. However, the ratios $D_i/D_j$ of the cell motilities are identifiable by Theorem~\ref{thm:global}. By Remark~\ref{rem:degen}, equal motilities alone do not remove identifiability. Among the there mentioned exceptional sets only full isotropy, with $A_s$, $D$ and $\Sigma$ all $\propto I$, involves equal motilities.

As an example, take tumor ($T$), cytotoxic effector ($E$) and macrophage ($M$) as three different cell types, and assume the interaction operator
\begin{equation}
A =
\begin{pmatrix}
-1 & -0.6 & 0.5\\
0.7 & -1.2 & -0.5\\
0.4 & 0.2 & -1.4
\end{pmatrix},
\qquad
D =
\begin{pmatrix}
1 & 0 & 0\\
0 & 2 & 0\\
0 & 0 & 1.4
\end{pmatrix},
\label{eq:worked}
\end{equation}
with the positive definite noise covariance
\begin{equation}
\Sigma =
\begin{pmatrix}
1 & 0.15 & -0.1\\
0.15 & 0.9 & 0.05\\
-0.1 & 0.05 & 1.2
\end{pmatrix}.
\label{eq:workedsigma}
\end{equation}
The off-diagonal signs are characteristic signs for the application context: effectors kill tumor cells ($A_{TE}<0$), tumor recruits effectors ($A_{ET}>0$), macrophages support tumor ($A_{TM}>0$), tumor recruits macrophages ($A_{MT}>0$), macrophages suppress effectors ($A_{EM}<0$). The operator is stable, its largest eigenvalue real part is $-0.874$, its symmetric part is negative definite, and a stationary field exists. The ratio of the two circulation directions $\lvert A_{TE}A_{EM}A_{MT}\rvert / \lvert A_{ET}A_{ME}A_{TM}\rvert$, is $1.71$: with $A_{ij}$ the effect of type $j$ on type $i$, the stronger loop is the tumor $\to$ macrophage $\to$ effector $\to$ tumor (numerator).

We now construct an alternative model that agrees with the above one at a single wavenumber. Take the transposed operator $A^{T}$, which reverses every directed interaction from above, keep the identical motility matrix $D$, and choose the noise covariance $\Sigma'$ such that the two models agree exactly at one wavenumber $\kappa_0$, which we assume here as $\kappa_0=1$. This is the linear equation $\Sigma' = -(M'X_0 + X_0 M'^{T})$ with $M'=A^{T}-\kappa_0 D$ and $X_0 := S(\kappa_0)$ the structure factor of the original model at $\kappa_0$. We have
\begin{equation}
A^T =
\begin{pmatrix}
-1 & 0.7 & 0.4\\
-0.6 & -1.2 & 0.2\\
0.5 & -0.5 & -1.4
\end{pmatrix},
\qquad
\Sigma' =
\begin{pmatrix}
0.895 & 0.255 & -0.074\\
0.255 & 1.008 & 0.029\\
-0.074 & 0.029 & 1.196
\end{pmatrix},
\end{equation}
\begin{equation}
X_0 =
\begin{pmatrix}
0.243 & 0.042 & 0.024\\
0.042 & 0.150 & 0.001\\
0.024 & 0.001 & 0.218
\end{pmatrix}.
\label{eq:worked_new}
\end{equation}
We have that $\Sigma'$ is symmetric and positive definite, and its eigenvalues are $0.68$, $1.19$ and $1.23$. For $\kappa_0$ the two models produce the same covariance matrix of the field and are indistinguishable from one another. This is an example of the degeneracy described in \cite{godreche2019,zhang2023}.

\begin{figure}[t]
\centering
\includegraphics[width=\textwidth]{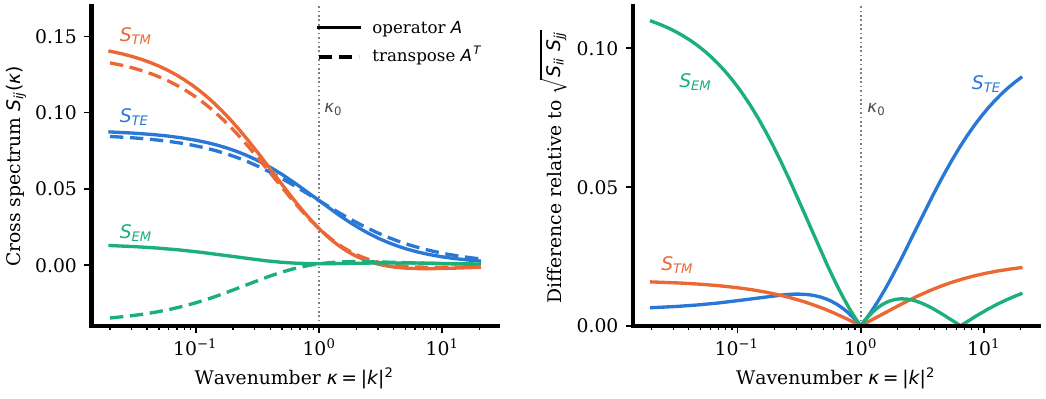}
\caption{The observation of only one covariance hides the interaction direction and cannot tell $A$ and $A^{T}$ apart, while the observation of the wavenumber family enables recoverability of direction. \textbf{Left:} the three cross spectra of the operator \eqref{eq:worked} (solid) and of its transpose (dashed) with the noise covariance refitted such that both agree at the observed single wavenumber $\kappa_0$ (dotted line). \textbf{Right:} the absolute difference between the two, in units of the spectrum's scale $\sqrt{S_{ii}S_{jj}}$. It vanishes at $\kappa_0$ by construction and grows differently on both sides, reaching about $11\%$ at the long-wavelength end and $9\%$ at the short end of the plotted range. This difference distinguishes the two models.}
\label{fig:familie}
\end{figure}

Figure~\ref{fig:familie} shows that the single covariance at $\kappa_0$ cannot tell $A$ and $A^{T}$ apart, i.e., directionality cannot be inferred from it, but that the full covariance family can unlock what happens off $\kappa_0$. The two models separate systematically in both directions, by about $11\%$ of $\sqrt{S_{ii}S_{jj}}$ at the long-wavelength end and about $9\%$ at the short end of the plotted range, against a difference that vanishes at $\kappa_0$ by construction. Since $S_{ij} \approx \Sigma _{ij}/(\kappa (D_I + D_j))$, towards large $\kappa$ the difference approaches a plateau that is independent of the operator and set by $\Sigma ' - \Sigma$. At every $\kappa$ the structure factor is real and symmetric (Lemma~\ref{lem:sym}), so the direction is defined by the variation with $\kappa$. Near $\kappa_0$ the two models are only slightly different, so an estimator restricted to a narrow band around $\kappa_0$ is close to degenerate irrespective of the number of cells, and the smallest resolvable wavenumber is set by the extent of the field. The figure shows that the information is in principle present in the model class, but not that it is present in a particular observed tissue section.

\section{Discussion}\label{sec:discussion}

By Theorem~\ref{thm:global}, the directed operator $A$ (in the example: killing, recruitment, suppression) and the motilities $D$ are globally identifiable from $\{S(\kappa):\kappa\ge0\}$ up to a positive time scale, at a generic point of the model class (Definition~\ref{def:generic}), for any number $K$ of cell types, and agreement at $\binom{K+1}{2}+1$ distinct wavenumbers already fixes the entire covariance family (Corollary~\ref{cor:shells}). The identified part of $A_a$ is its half-line, i.e., the sign pattern and the ratios of its entries. Under the per-cell-type normalization, the signs, zeros of $A$ and the products of entries along closed paths survive, and the $K$ coupling scales are lost (Corollary~\ref{cor:norm}). The direction can be inferred from $S(k)$ through the dependence of the cross- and autospectra on $\kappa$, but not through the mean densities.
These are independent of $(A,D,\Sigma)$ in the linearized model. The cross spectra $S_{ij}(\kappa)=S_{ji}(\kappa)$ are symmetric in the pair, and together with the autospectra their family identifies $A_a$. Descriptors that translate the pattern to a $K\times K$ matrix (such as the above mentioned neighborhood enrichment or co-localization) at one scale can determine at most $K^2$ numbers for the $n_U-1$ parameters after the time scale. From this work, it remains unknown whether these matrices can infer directionality.
How the estimate's conditioning depends on the motility differences is not investigated here. Further, with unknown $\Sigma$, no snapshot observation can fix the absolute rates (Proposition~\ref{prop:scale}), but external clocks can provide this scale. These are, for example, a second time point, or a division marker (with calibrated decay after cell-cycle exit \citep{miller2018}).
Parameters of calibrated mechanistic models have served as patient-specific markers of outcome, both with a clock available \citep{benzekry2024} and, for a single population, on a normalized time scale \citep{schlicke2026hdtd}. Analogous use of the half-line of $A$ does not require a clock, but statements about absolute rates do. A baseline section observation therefore constrains $A$ and $D$ up to the time scale, i.e., the linearized interaction structure about the stationary state, but not the nonlinear dynamics beyond. However, the latter is required to determine basins of attraction or separatrices. Two groups that share the same baseline structure but differ in absolute rates are indistinguishable by construction. A second section observation or a proliferation anchor can remove this ambiguity.

The entries of $A$ should not be confused with rate constants. Linearization of a mass-action term $-c\rho_T \rho_E$ around $(\lambda_T,\lambda_E)$ gives $A_{TE}=-c\lambda_T$, so two tissues with identical biology but different composition have a different $A$. Wherever the entries of $A$ are compared across patients or cohorts, $A_{ij}/\lambda_i$ should be reported alongside $A_{ij}$ for mass-action interactions. The ratio of the two circulation directions is invariant under this correction and under normalization.

The results relate to several previous works. The bound that a single stationary covariance leaves the antisymmetric part of a linear drift undetermined \citep{godreche2019,zhang2023} is not absolute even for one covariance. For a drift whose directed graph contains no directed $2$-cycles, the drift is globally identifiable from one stationary covariance \citep{dettling2023}. For a connected graph with a self-loop at every node, non-Gaussian stationary data identify the drift generically through higher order cumulants if the second-order and one higher-order cumulant of the noise are diagonal, again up to a common positive factor \citep{recke2026}. Sparsity enables the same, for example as in the estimation of directed brain connectivity from zero-lag covariances by $L^1$-minimization \citep{schiefer2018}. None of these routes is available for our context: tumor-immune interactions contain reciprocal pairs and dense blocks, thus contain $2$-cycles, and the result of \cite{dettling2023} requires a known volatility matrix. Within continuous Lyapunov models, model equivalence for acyclic graphs with uncorrelated noise is finer than Markov equivalence \citep{amendola2025}, and edge-sign identifiability has been studied for known causal structure, diagonal diffusion and a single stationary covariance \citep{vanseeventer2026}. Relative to this baseline, Theorem~\ref{thm:main} highlights the gain of the spatial degree of freedom: with the family $\{S(\kappa)\}$ in place of one covariance, no graph assumed and $D$ unknown, the directed operator is recovered as a whole.

Single-snapshot reconstruction has been done for birth-diffusion-death processes, where only dimensionless ratios are identifiable \citep{miles2023}, for pattern-forming reaction-diffusion models from one partially observed image \citep{kazarnikov2025}, or for scalar parameter fields of spatial stochastic dynamics \citep{gu2026}. The structural identifiability of the fully observed linear-in-parameter parabolic equations has been analyzed through elliptic operators \citep{salmaniw2024}. For a three-state velocity-jump process the parameters are globally identifiable from single-particle trajectories, but only locally identifiable from the total particle density under distinct velocities \citep{ceccarelli2026}. In a neighboring application, cell-population dynamics have been reconstructed from one spatial-proteomics section by identifying division rates as functions of the neighborhood compositions of cells \citep{somer2026}. This work fitted a predictive model of how the composition evolves over time, whereas the objective here is the identification of a directed interaction operator along with the data requirements to determine it.

Recently, Celcomen identifiably retrieved the undirected coupling of a Gibbs model on a fixed spatial adjacency and the authors state that only the Markov equivalence class of the underlying graph can be recovered without interventional data and without further assumptions \citep{megas2026}. The bound concerns nonparametric i.i.d. observational data from a static graph. The additional physical structure utilized in the present work is a non-equilibrium stationary state of a spatially extended stochastic dynamics, and the observable is the full covariance family $\{S(\kappa)\}$. For known $(D,\Sigma)$, $A$ is not identifiable when nonzero $J$ results in antisymmetric $JS(k)$ for all $k$ by Corollary~\ref{cor:id}, and for unknown $(D,\Sigma)$ Theorem~\ref{thm:main} and Remark~\ref{rem:degen} show that a kernel beyond the time scale happens only in very specific cases. The difference between the two approaches is the modeling assumptions, which can be tested on data.

The structure factor of a reaction-diffusion Langevin field for a single population $K=1$ is textbook material \citep{gardiner}, and its explicit form has been given repeatedly \citep{brito2007,kim2017}. For several interacting species, fluctuation spectra in the linear noise approximation are utilized in the analysis of stochastic Turing patterns \citep{biancalani2010,mckane2014}. Here, we focused on the inverse problem for multiple interacting populations. The identifiability proved here is relative to this model class, $A$ does not depend on $k$, the wavenumber-dependence is through $-|k|^2D$, and white noise with $\Sigma$ independent of the wavenumber $k$. Without these restrictions, any realizable $g_2(r)$ or $S(k)$ of a translationally invariant non-equilibrium system is conjectured to be reproducible by an equilibrium ensemble with effective pair interactions \citep{torquato2022,zhangtorquato2020}, and the equilibrium uniqueness results concern a symmetric pair potential of a single component determined by the pair correlation function \citep{frommer2019,ariel2020}.

Important limitations remain. First, the theory is linear and describes the linearization around a stationary state. The bistable control-versus-escape structure of tumor-immune dynamics \citep{kuznetsov1994,dunn2004}, on which calibrated immunoediting models forecast checkpoint-inhibitor outcomes \citep{scibilia2024}, lies beyond this linear core. The identifiability results do not extend to it, and forward simulations of calibrated models for counterfactual treatment comparisons, as have been used in preclinical and clinical settings \citep{schlicke2026bjc,schlicke2026hdtd}, require the nonlinear model. However, such an extension builds on the identifiability of the linearized operator. Second, the field is assumed to be stationary, whereas a growing tumor is transient. The results apply to regions and time windows in which the stationary state is approximately valid. The small wavenumbers contain the information for $A$, and their relaxation times are expected to be the highest. This means the identifiable and the stationary band are expected to point in opposite directions. Third, the field is treated as a linear Gaussian one. The second order results do not use the Gaussianity but the linear field serves as an approximation to a discrete cellular system and enlarging the field at fixed density is not expected to improve it.

Several unaddressed questions remain. A wavenumber-dependent noise covariance is not treated in general. For cell populations, motion is a conserved process that contributes $2 D_i \lambda_i |k|^2$ to the noise of population $i$ \citep{kim2017}. On the other hand, birth and death only contribute a constant, and $\Sigma(k) = \Sigma_0 + 2 \kappa D \Lambda$ with $\Lambda=\mathrm{diag}(\lambda_1,\dots,\lambda_K)$, and $S(k) \to \Lambda$ for large $\kappa$, which is exactly the removed Poisson term. The Poisson-corrected family $S-\Lambda$ solves \eqref{eq:lyap} with the covariance $\Sigma_0+A\Lambda+\Lambda A^{T}$ which is independent of the wavenumber. This covariance does not need to be positive definite. The identifiability results (Lemma~\ref{lem:finite} to Corollary~\ref{cor:norm}) only use Assumption~\ref{ass:stab} and the symmetry of $\Sigma$. Positivity of $\Sigma$ is only needed in Lemma ~\ref{lem:sym}, Proposition~\ref{prop:inv} and Corollary~\ref{cor:id}, and for $s=0$ in Lemma~\ref{lem:inverseimage} it is replaced with $A=0$ not being Hurwitz. Thus, they cover the entire Poisson-corrected family.
Further, the application to tumor-immune tissue takes a dimension of $d=2$, but histological sections are planes cut through three-dimensional tissues. The structure factor of such a cut is the three-dimensional integration over the wavenumber normal to the section, with a different $\kappa$-profile and a tail that decays with $|k|^{-1}$ rather than $|k|^{-2}$. With space-time white noise the solution of \eqref{eq:rde} is not function-valued for $d\ge2$ and $C(0)$ diverges, logarithmically for $d=2$ and as a power for $d=3$, so a cutoff at the cell scale is necessary. A directed transport $V \cdot \nabla u(x,t)$ with a cell-type-dependent velocity $V$ makes $S(k)$ complex Hermitian and anisotropic. Lemma~\ref{lem:sym} then no longer holds. This results in a different identifiability problem \citep{browning2023,ceccarelli2026}, deliberately excluded here by the model choice. Continuous Lyapunov models commonly assume a diagonal $\Sigma$ \citep{dettling2023,recke2026}, whereas a flux between two cell types, for example progenitor to terminally exhausted effector or monocyte to macrophage, results in anticorrelated demographic noise with $\Sigma_{ij}<0$. The theorems admit any symmetric positive definite $\Sigma$, and the reference point in the proof of Theorem~\ref{thm:main} has $\Sigma$ with every entry nonzero. This means that estimators for these pairs may be misspecified, whereas they are covered by the identifiability results.

Corollary~\ref{cor:norm} has one unknown factor per population. Normalizations that depend on the scale, the finite and masked observation window, and the two-dimensional section through three-dimensional tissue are left open. Further, the theorems have to be extended for the third-order statistics, which are expected to recover the interaction scales when their source is linked to the same rates as the operator (which is the case e.g. for birth and death). Without this extension or a constraint on $\Sigma$ that is not invariant under $\Sigma \mapsto G \Sigma G$, the coupling remains open whenever the detection probabilities are unknown.

\backmatter

\section*{Declarations}

\textbf{Funding.}
Pirmin Schlicke was funded within the APART-USA program of the Austrian Academy of Sciences (OeAW).\\
\textbf{Competing interests.}
The author declares no competing interests.\\
\textbf{Data availability.}
No datasets were generated or analyzed in this article.\\
\textbf{Code availability.}
The code necessary to reproduce Figure \ref{fig:familie} and the rank computations of Remark \ref{rem:degen} are available from the author.

\bibliography{sn-bibliography}

\end{document}